\documentclass[11pt]{amsart}

\pdfoutput=1

\usepackage[text={420pt,660pt},centering]{geometry}

\usepackage{cancel}
\usepackage{esint,amssymb}
\usepackage{graphicx}
\usepackage{mathtools} 
\usepackage[colorlinks=true, pdfstartview=FitV, linkcolor=blue, citecolor=blue, urlcolor=blue,pagebackref=false]{hyperref}
\usepackage{microtype}

\usepackage{bm}
\usepackage{dsfont}
\usepackage{mathrsfs}
\usepackage{xcolor}
\usepackage{accents} \usepackage{enumitem} \parskip= 2pt
\usepackage[normalem]{ulem}

\newtheorem{proposition}{Proposition}
\newtheorem{theorem}[proposition]{Theorem}
\newtheorem{lemma}[proposition]{Lemma}
\newtheorem{corollary}[proposition]{Corollary}

\theoremstyle{remark}
\newtheorem{remark}[proposition]{Remark}

\theoremstyle{definition}

\theoremstyle{plain}
\newtheorem*{remark*}{Remark} 

\numberwithin{equation}{section}
\numberwithin{proposition}{section}
\numberwithin{figure}{section}
\numberwithin{table}{section}

\newcommand{\N}{\mathbb{N}}

\newcommand{\R}{\mathbb{R}}

\newcommand{\E}{\mathbb{E}}

\renewcommand{\S}{\mathbf{S}}

\newcommand{\soc}{{\mathsf{soc}}}

\newcommand{\eps}{\varepsilon}

\renewcommand{\le}{\leqslant}
\renewcommand{\ge}{\geqslant}
\renewcommand{\leq}{\leqslant}
\renewcommand{\geq}{\geqslant}

\renewcommand{\subset}{\subseteq}
\renewcommand{\bar}{\overline}
\renewcommand{\tilde}{\widetilde}

\newcommand{\Ll}{\left}
\newcommand{\Rr}{\right}
\renewcommand{\d}{\mathrm{d}}

\newcommand{\D}{D}

\DeclareMathOperator{\supp}{supp}

\newcommand{\la}{\left\langle}
\newcommand{\ra}{\right\rangle}

\newcommand{\cF}{\mathcal{F}}

\newcommand{\sP}{\mathscr{P}}

\newcommand{\cK}{\mathcal{K}}
 
\newcommand{\bfs}{\mathbf{s}}

\newcommand{\rv}[1]{\textcolor{black}{#1}}

\begin{document}

\author{Hong-Bin Chen}
\address{Institut des Hautes \'Etudes Scientifiques, France}
\email{hbchen@ihes.fr}


\title{Free energy in spin glass models with conventional order}

\begin{abstract}
Recently, \cite{Baldwin2023} considered spin glass models with additional conventional order parameters characterizing single-replica properties. These parameters are distinct from the standard order parameter, the overlap, used to measure correlations between replicas. A ``min-max'' formula for the free energy was prescribed in~\cite{Baldwin2023}. We rigorously verify this prescription in the setting of vector spin glass models featuring additional deterministic spin interactions. Notably, our results can be viewed as a generalization of the Parisi formula for vector spin glass models in \cite{pan.vec}, where the order parameter for self-overlap is already present.
\end{abstract}

\maketitle

\section{Introduction}\label{s.intro}

It is well understood that the classical Sherrington--Kirkpatrick (SK) model has one order parameter that characterizes the correlations between replicas of the system\rv{, which is called the \textit{overlap} \cite{parisi1983order}}. In more general models, there can be other order parameters. For instance, in the vector spin glass, the self-overlap of a single replica comes into play. We call order parameters that characterize the properties of a single replica the \textit{conventional} order parameter. In models where both types of order exist, \cite{Baldwin2023} recently proposed a ``min-max'' prescription for the variational formula of free energy. We rigorously verify this in the setting of vector spin glass.

More precisely, we consider a Hamiltonian of the form $H_N(\sigma)+G_N(\sigma)$ at size $N$. Here, $H_N(\sigma)$ is the standard spin glass Hamiltonian with Gaussian disorder, and $G_N(\sigma)$ accounts for additional deterministic spin interactions. \rv{Also, $G_N$ can be very general so that $H_N(\sigma)+G_N(\sigma)$ cannot be rewritten as $H_N(\sigma)$ with shifted coupling coefficients.} Let $\pi$ be the overlap parameter and $m$ be the conventional order parameter in this model. We show that the limit of free energy is given by
\begin{align*}
    \lim_{N\to\infty} \frac{1}{N}\E \log \int e^{H_N(\sigma)+G_N(\sigma)} \d \sigma = \rv{\sup_m\inf_\pi \mathcal{P}(\pi,m)}
\end{align*}
for some Parisi-type functional $\mathcal{P}$, verifying \cite[Eq.\ (12)]{Baldwin2023} \rv{(there, the free energy has an additional minus sign; so ``min-max'' there corresponds to ``max-min'' here)}.
\rv{We also clarify that the quantity on the left-hand side is usually called the ``pressure'', which differs from the ``free energy'' by a multiplicative factor of the temperature (absorbed into $H_N(\sigma)$ and kept implicit here). Here, for practical purposes, we call it ``free energy'' by slightly abusing the language.}

Formulas of this form have already appeared in the generalized SK model~\cite{pan05}, the Potts spin glass~\cite{pan.potts}, and general vector spin glass~\cite{pan.vec}. A common feature is that the normalized self-overlap is not constant (in contract, this quantity is constantly equal to $1$ in the SK model). In these models, $m$ characterizes the self-overlap. 

Our result allows for more general $m$. For instance, $m$ can be the parameter for the mean magnetization, moments of spins, the self-overlap, and combinations of these quantities. Regardless of the choice, $m$ has always to be optimized after the overlap~$\pi$.

Our proof is based on interpolating between the free energy with Hamiltonian $H_N(\sigma)+G_N(\sigma)$ and the free energy with Hamiltonian $H_N(\sigma) - \frac{1}{2}\E H_N(\sigma)^2$ along a Hamilton--Jacobi equation. The simplest form of this technique exists for the Curie--Weiss (CW) model (see \cite[Chapter~3]{HJbook}). Variants of this have been used to handle the self-overlap order parameter in \cite{mourrat2020extending,chen2023self}. We will see that the free energy with the latter Hamiltonian is ``pure'' in the sense that the only order parameter is the overlap, the same as the SK model.

The approach in \cite{pan05,pan.potts,pan.vec} is based on considering free energy with spins constrained on a subset to have certain self-overlap. We believe that this method can be modified to prove the results here by considering constraints for more general conventional order parameters. But, to do so, one seemingly has to rework the argument from the very beginning. In comparison, the PDE approach is more modular and easily applicable. Nevertheless, we will present a proof via this alternative approach in a special case, where only minimal modifications of results in~\cite{pan.vec} are needed.

\subsection{Setting}

We work with vector-valued spins distributed independently and identically. Let $D\in\N$ be the dimension of a single spin and let
\begin{enumerate}[label={\rm (H0)}]
    \item \label{i.P_1} $P_1$ be a finite measure supported on the closed unit ball in $\R^\D$.
\end{enumerate}
For $N\in\N$, we denote the spin configuration by $\sigma = (\sigma_{ki})_{1\leq k\leq\D,\, 1\leq i\leq N}$ which is a $\D\times N$ matrix in $\R^{\D\times N}$. We view each column vector $\sigma_i = (\sigma_{ki})_{1\leq k\leq \D}$ in $\sigma$ as an $\R^\D$-valued vector spin. We sample each vector spin independently from $P_1$ and denote the distribution of $\sigma$ by $P_N$. More precisely, we have $\d P_N(\sigma) = \otimes^N_{i=1}\d P_1(\sigma_i)$.

For each $N$, we are given a centered real-valued Gaussian process $(H_N(\sigma))_{\sigma\in\R^{\D\times N}}$ with covariance
\begin{align*}\E H_N(\sigma) H_N(\sigma') = N \xi\Ll(\frac{\sigma\sigma'^\intercal}{N}\Rr),\quad\forall \sigma,\,\sigma'\in \R^{\D\times N}
\end{align*}
for some deterministic function $\xi:\R^{\D\times \D}\to \R$ satisfying conditions~\ref{i.xi_loc_lip}--\ref{i.xi_convex} specified later in Section~\ref{s.proofs}. In particular, we require $\xi$ to be convex over the set of $\D\times\D$ positive semi-definite matrices.

For each $N\in\N$, the standard free energy is
\begin{align}\label{e.F_N=}
    F_N = \frac{1}{N}\E \log \int \exp\Ll(H_N(\sigma)\Rr)\d P_N(\sigma).
\end{align}
We also consider a version of free energy that is free of conventional order parameters. For each $N\in\N$, the free energy with self-overlap correction is
\begin{align}\label{e.F^soc_N=}
    F_N^\soc = \frac{1}{N}\E \log \int \exp\Ll(H_N(\sigma) - \frac{N}{2}\xi\Ll(\frac{\sigma\sigma^\intercal}{N}\Rr) \Rr)\d P_N(\sigma).
\end{align}
We call the term $- \frac{N}{2}\xi\Ll(\frac{\sigma\sigma^\intercal}{N}\Rr)$ the \textit{self-overlap correction}, which is is equal to $-\frac{1}{2}\E H_N(\sigma)^2$ and resembles the drift term in an exponential martingale.
It has been proved in \cite[Corollary~8.3]{HJ_critical_pts} that
\begin{align}\label{e.limF^soc_N=}
    \lim_{N\to\infty}F_N^\soc = \inf_{\pi\in\Pi}\sP(\pi)
\end{align}
where the Parisi-type functional $\sP$ is defined later in~\eqref{e.sP(pi)=}. Here, $\Pi$ is the collection of overlap parameters and there is no conventional order parameter.

Next, we introduce additional spin interactions.
Fix any $d\in \N$ and let
\begin{enumerate}[label={\rm (H5)}]
    \item \label{i.h_G} $h:\R^\D\to \R^d$ be bounded  and measurable, and $G:\R^d\to\R$ be locally Lipschitz.
\end{enumerate}
We do not require $G$ to be convex.
We let $h$ act on each spin and we can view $h(\sigma_i)$ as a new spin distributed according to the pushforward of $P_1$ under $h$. We denote the mean magnetization of these new spins by
\begin{align}\label{e.m=}
    m_N = \frac{1}{N}\sum_{i=1}^N h(\sigma_i).
\end{align}
We consider the free energy\begin{align}\label{e.F^soc,G_N=}
    F_N^{\soc,G} = \frac{1}{N}\E \log \int \exp\Ll(H_N(\sigma) - \frac{N}{2}\xi\Ll(\frac{\sigma\sigma^\intercal}{N}\Rr) + NG\Ll(m_N\Rr)\Rr)\d P_N(\sigma).
\end{align}
The additional term $NG(m_N)$ is of the type in the generalized CW model. The natural order parameter arising from $NG(m_N)$ characterizes the limit of $m_N$.

Removing the correction term, we also consider
\begin{align}\label{e.F^G_N=}
    F_N^G = \frac{1}{N}\E \log \int \exp\Ll(H_N(\sigma)  + NG\Ll(m_N\Rr)\Rr)\d P_N(\sigma).
\end{align}
For this, the self-overlap comes into play and we need to introduce the following. Let $\S^\D$ be the set of $\D\times\D$ real symmetric matrices. For every $a,b\in\S^\D$, we write $a\cdot b= \sum_{ij}a_{ij}b_{ij}$. By fixing an orthogonal basis for $\S^\D$ under this inner product, we can identify $\S^\D$ with $\R^{D(D+1)/2}$ isometrically. Let $\bfs:\R^\D\to \S^\D$ be given by
\begin{align}\label{e.bfs}
    \bfs(\tau) = \tau\tau^\intercal,\quad\forall \tau\in \R^\D,
\end{align}
and notice that the self-overlap is given by $\frac{\sigma\sigma^\intercal}{N} = \frac{1}{N}\sum_{i=1}^N \bfs(\sigma_i)$.

For any $x$ and $y$ in the same Euclidean space (e.g.\ $\R^d$, $\S^\D$), we denote the inner product between them by $x\cdot y$ and we write $|x|=\sqrt{x\cdot x}$.

\subsection{Main results}

\begin{theorem}\label{t.F^soc,G_N}
Under conditions~\ref{i.P_1}--\ref{i.h_G}, the limit of $F^{\soc,G}_N$ in~\eqref{e.F^soc,G_N=} is given by
\begin{align}\label{e.limF^soc,G_N=}
    \lim_{N\to\infty} F^{\soc,G}_N = \sup_{m\in\R^d}\inf_{\pi\in\Pi}\inf_{x\in\R^d}\Ll\{\sP^h(\pi,x)-m\cdot x+G(m)\Rr\}.
\end{align}
\end{theorem}
Here, the functional $\sP^h(\pi,x)$ is of the Parisi type, whose explicit expression is in~\eqref{e.sP(pi,x)}. For each $\pi\in\Pi$, we consider the convex conjugate
\begin{align}\label{e.P^*=}
    {\sP^h}^*(\pi, m) = \sup_{x\in\R^d} \Ll\{m\cdot x - \sP^h(\pi,x)\Rr\},\quad\forall m\in \R^d.
\end{align}
In this notation, we can rewrite~\eqref{e.limF^soc,G_N=} as
\begin{align*}
    \lim_{N\to\infty} F^{\soc,G}_N = \sup_{m\in\R^d}\inf_{\pi\in\Pi}\Ll\{-{\sP^h}^*(\pi,m)+G(m)\Rr\}
\end{align*}
which recovers the min-max formula prescribed in \cite{Baldwin2023} (see (12) therein where the free energy is defined to have a minus sign in the front; so ``min-max'' there corresponds to ``max-min'' here).

By incorporating $\bfs$ into $h$ and $\xi$ into $G$, we can remove the self-overlap correction.
\begin{corollary}\label{c.F^G_N}
Under conditions~\ref{i.P_1}--\ref{i.h_G}, the limit of $F^{G}_N$ in~\eqref{e.F^G_N=} is given by
\begin{align}\label{e.limF^G_N=}
    \lim_{N\to\infty} F^G_N = \sup_{(z,m)\in\S^\D\times\R^d}\inf_{\pi\in\Pi}\inf_{x\in \S^\D\times \R^d}\Ll\{\sP^{(\bfs,h)}(\pi,x)-(z,m)\cdot x +\frac{1}{2}\xi(z)+G(m)\Rr\}.
\end{align}
\end{corollary}

Here, $\sP^{(\bfs,h)}$ is defined as $\sP^h$ with $h$ substituted with the function $\S^\D\times \R^d\ni(z,m)\mapsto (\bfs(z),h(m))$.
Again, by absorbing $\inf_x$ into a convex conjugate as in~\eqref{e.P^*=}, we recover the min-max prescription in~\cite{Baldwin2023}.
Here, without the self-overlap correction, it is necessary to include the self-overlap into the conventional order parameter. 

When $h$ contains $\mathbf{s}$, one can simplify the formula in~\eqref{e.limF^G_N=} by modifying the proof. A particular case is when $h=\mathbf{s}$ and we have the following.

\begin{corollary}\label{c.F^G_N_so}
Under conditions~\ref{i.P_1}--\ref{i.h_G} and an additional assumption that $h=\bfs$ (identifying $\S^\D$ with $\R^{\D(\D+1)/2}$ isometrically), the limit of $F^{G}_N$ in~\eqref{e.F^G_N=} is given by
\begin{align}\label{e.limF^G_N=so}
    \lim_{N\to\infty} F^G_N = \sup_{z\in\S^\D}\inf_{\pi\in\Pi}\inf_{x\in \S^\D}\Ll\{\sP^{\bfs}(\pi,x)-z\cdot x +\frac{1}{2}\xi(z)+G(z)\Rr\}.
\end{align}
In particular, the limit of $F_N$ in~\eqref{e.F_N=} is given by
\begin{align}\label{e.limF_N=}
    \lim_{N\to\infty} F_N = \sup_{z\in\S^\D}\inf_{\pi\in\Pi}\inf_{x\in \S^\D}\Ll\{\sP^{\bfs}(\pi,x)-z\cdot x +\frac{1}{2}\xi(z)\Rr\}.
\end{align}
\end{corollary}
Here, \eqref{e.limF_N=} is obtained from~\eqref{e.limF^G_N=so} by setting $G=0$ and \eqref{e.limF_N=} recovers the result for the standard vector spin glass in~\cite{pan.vec}. 

Section~\ref{s.proofs} is devoted to the proofs of these results using PDE techniques. In Section~\ref{s.alt_proof}, we present an alternative proof of Corollary~\ref{c.F^G_N_so} using the argument of constrained free energy in~\cite{pan05,pan.potts,pan.vec}, which is closer to the analysis done in~\cite{Baldwin2023}.

\begin{remark}\label{r.m_N}
The free energy $F^\soc_N$ in~\eqref{e.F^soc_N=} is ``pure'' since the overlap is the only order parameter in the formula. Further, we will show in Proposition~\ref{p.cvg_m} of Appendix~\ref{s.cvg_m} that, with the self-overlap correction, all conventional orders are trivial in the sense that $m_N$ converges to a constant for any bounded $h$.
\end{remark}

\begin{remark}In view of Corollary~\ref{c.F^G_N}, without the self-overlap correction, it is necessary to include the self-overlap as part of the conventional order. From~\eqref{e.limF_N=}, it is also a minimal requirement. One can interpret this as that the self-overlap is the canonical conventional order parameter. This is not surprising because the self-overlap together with the overlap (which characterizes the correlation between different replicas) completely describes the entire overlap array.
\end{remark}

\begin{remark}
\rv{We comment on the convexity requirement~\ref{i.xi_convex} on $\xi$. In \cite{Baldwin2023}, the authors stated below \cite[Eq.~(77)]{Baldwin2023} that the convexity assumption is needed for technical reasons and that the min-max formula should hold without it. However, we disagree with this for the following reasons. It is argued in \cite[Section~6]{mourrat2021nonconvex} that the free energy in the bipartite model (where $\xi$ is not convex) does not admit a variational representation, or at least one that originates from rewriting the Parisi formula as a saddle point problem. Hence, it is not clear if, in general, there is a variational formula to begin with, when there is no additional conventional order.}
\end{remark}

\subsection{Related works}

As mentioned before, the source of motivation is from~\cite{Baldwin2023}, which used the ``min-max'' prescription to clarify the relationship between quenched free energies, annealed free energies, and the overlap (which is called the ``replica order parameter'' therein). An earlier work that considered a more restrictive class of spin glass models with conventional order parameters is~\cite{mottishaw1986first}. Since the self-overlap is one particular conventional order parameter, these connect with the author's recent studies~\cite{chen2023self,chen2023on} on the self-overlap in vector spin glasses. 
The idea of adding the self-overlap correction first appeared in the Hamilton--Jacobi approach to spin glass models \cite{mourrat2022parisi,mourrat2020extending,mourrat2021nonconvex,mourrat2023free}. For more detail on this approach, we refer to~\cite{HJbook}.

\rv{Works on spin glasses with CW-type interactions include~\cite{chen2014mixed,camilli2022inference}. The model in~\cite{chen2014mixed} has the mixed even-order SK Hamiltonian with additional term $\frac{\beta }{2N}\Ll(\sum_{i=1}^N\sigma_i\Rr)^2 + \sum_{i=1}^Nh_i\sigma_i$ with $\beta\geq 0$ and i.i.d.\ Gaussians $(h_i)_{1\leq i\leq N}$. The variational formula for the free energy is obtained through a simple argument via usual approaches from the CW model. In~\cite{camilli2022inference}, the authors considered the SK Hamiltonian plus $\frac{\nu }{2N}\Ll(\sum_{i=1}^N\xi_i\sigma_i\Rr)^2 + \lambda\sum_{i=1}^N\xi_i\sigma_i$ with $\nu\geq0$, $\lambda\in\R$, and i.i.d.\ random variables $(\xi_i)_{1\leq i\leq N}$ with finite fourth moment. The proof relies on the adaptive interpolation method introduced in~\cite{barbier2019adaptive,barbier2019adaptive2}. The model in~\cite{chen2014mixed} can be covered here by absorbing the centered part of $\sum_{i=1}^Nh_i\sigma_i$ into $H_N(\sigma)$ and we believe that the method demonstrated there does not work if the CW interaction is non-convex. While our results are not directly applicable to the model in~\cite{camilli2022inference}, we believe that a modified argument similar to the one for statistical inference models in~\cite{chen2022statistical} could work.}

Since the theme of this work is to find a formula for free energy, we review works along this direction. Parisi initially proposed the formula for the free energy in the SK model in \cite{parisi79,parisi80}. Guerra rigorously proved the upper bound in \cite{gue03} and Talagrand proved the matching lower bound in \cite{Tpaper}.
Extensions were made to various settings: the SK model with soft spins \cite{pan05}, the scalar mixed $p$-spin model \cite{panchenko2014parisi,pan}, the multi-species model \cite{barcon,pan.multi}, and the mixed $p$-spin model with vector spins \cite{pan.potts,pan.vec}. The Parisi formula for the balanced Potts spin glass was recently established in~\cite{bates2023parisi}.
For spherical spins, Parisi-type variational formulas have been proved for the SK model \cite{tal.sph}, the mixed $p$-spin model \cite{chen2013aizenman}, and the multi-species model \cite{bates2022free}.

\rv{
As we are employing PDE tools in the study of spin glasses, we briefly describe the history of this approach. It was initiated by Guerra in~\cite{guerra2001sum} in the replica symmetric regime, which was later extended to the replica symmetry breaking scenario in~\cite{barra1} and explored further in various settings in~\cite{barra2,abarra,barramulti,barra2014quantum}. The application of the PDE techniques to the CW model starts in~\cite{barra2008mean} and continues in~\cite{genovese2009mechanical}. Mathematically, the PDE approach was adopted by Mourrat in~\cite{mourrat2022parisi,mourrat2020hamilton}. Applications to statistical inference models also include~\cite{mourrat2021hamilton,HB1,HBJ,chen2022statistical,chen2023free}. The well-posedness of the PDEs arising from these contexts has been studied in~\cite{chen2022hamilton,chen2022hamilton2}.
}

\subsection{Acknowledgements}
The author thanks Jean-Christophe Mourrat for stimulating discussions. This project has received funding from the European Research Council (ERC) under the European Union’s Horizon 2020 research and innovation programme (grant agreement No.\ 757296).

\section{Proofs}\label{s.proofs}

The goal is to prove our main results. We start with some definitions and conditions on $\xi$.
Let $\S^\D_+$ denote the subset of $\S^\D$ comprising of positive semi-definite matrices. 
A useful fact (see \cite[Theorem~7.5.4]{horn2012matrix}) is that if $a \in \S^\D$, then
\begin{align}\label{e.a.b>0}
a\in\S^\D_+ \quad \Longleftrightarrow\quad \forall b\in \S^\D_+:\ a\cdot b\geq 0.
\end{align}
Throughout, we assume that
\begin{enumerate}[start=1,label={\rm (H\arabic*)}]
    \item \label{i.xi_loc_lip}
    $\xi$ is differentiable and $\nabla\xi$ is locally Lipschitz;
    \item \label{i.xi_sym} $\xi\geq 0$ on $\S^\D_+$, $\xi(0)=0$, and $\xi(a)=\xi(a^\intercal)$ for all $a\in \R^{\D\times\D}$;
    \item \label{i.xi_incre} if $a,\,b\in\S^\D_+$ satisfies $a - b\in\S^\D_+$, then $\xi(a)\geq \xi(b)$ and $\nabla \xi(a)- \nabla \xi(b)\in \S^\D_+$.
    \item \label{i.xi_convex}
    $\xi$ is convex over $\S^\D_+$.
\end{enumerate}
The derivative $\nabla \xi:\R^{\D\times \D}\to \R^{\D\times \D}$ is defined with respect to the entry-wise inner product.
This encompasses a broad class of models and we refer to~\cite[Section~6]{mourrat2023free} for examples.

\subsection{Parisi-type functional}

We recall the definition of the Parisi-type functional appearing in the variational formulas.
The collection $\Pi$ consists of matrix-valued paths:
\begin{align*}\Pi = \Ll\{\pi:[0,1]\to\S^\D_+\  \big|\  \text{$\pi$ is left-continuous and increasing}\Rr\}
\end{align*}
where 
$\pi$ is increasing in the sense that
\begin{align*}
    s'\geq s\quad\implies\quad \pi(s')- \pi(s)\in\S^\D_+.
\end{align*}
To define the Parisi functional, we recall the \textit{Ruelle probability cascade} \cite{ruelle1987mathematical}. Let $\mathfrak{R}$ denote the Ruelle probability cascade with overlap uniformly distributed over $[0,1]$ (see \cite[Theorem~2.17]{pan}). Precisely, $\mathfrak{R}$ is a random probability measure on the unit sphere in a separable Hilbert space, with the inner product denoted by $\alpha\wedge\alpha'$. Independently sampling $\alpha$ and $\alpha'$ from $\mathfrak{R}$, the law of $\alpha\wedge\alpha'$ under $\E \mathfrak{R}^{\otimes 2}$ is the uniform distribution over $[0,1]$, where $\E$ integrates the randomness inherent in $\mathfrak{R}$. Almost surely, the support of $\mathfrak{R}$ is ultrametric in the induced topology. For rigorous definitions and comprehensive properties, we direct the reader to \cite[Chapter2]{pan} (also see \cite[Chapter5]{HJbook}).

Conditioned on $\mathfrak{R}$ (i.e.\ fixing any of its realization), for each $\pi\in\Pi$, let $(w^\pi(\alpha))_{\alpha\in\supp\mathfrak{R}}$ be a centered $\R^\D$-valued Gaussian process with covariance
\begin{align*}\E w^\pi(\alpha)w^\pi(\alpha')^\intercal = \pi (\alpha\wedge\alpha'),\quad\forall \alpha,\alpha\in\supp\mathfrak{R}.
\end{align*}
For the construction and properties of this process, we refer to \cite[Section~4 and Remark~4.9]{HJ_critical_pts}.

Based on the last property in \ref{i.xi_sym}, it follows that $\nabla \xi(a)\in\S^\D$ when $a \in \S^\D$. Consequently, the combination of the first property in \ref{i.xi_incre} and \eqref{e.a.b>0} ensures that $\nabla\xi(a)\in\S^\D_+$ for every $a\in\S^\D_+$. Leveraging these observations along with the second property in \ref{i.xi_incre}, we deduce that $\nabla\xi\circ\pi\in\Pi$ holds for all $\pi\in\Pi$.

With these elements in place, we define the Parisi functional $\sP^h(\pi,x)$ for $\pi\in\Pi$ and $x\in\S^\D$ as follows:
\begin{align}\label{e.sP(pi,x)}
\begin{split}
    \sP^h(\pi,x)= \E \log\iint \exp\left(w^{\nabla\xi\circ\pi}(\alpha)\cdot \tau - \frac{1}{2}\nabla\xi\circ\pi(1)\cdot \tau\tau^\intercal+x\cdot h(\tau)\right)\d P_1(\tau)\d \mathfrak{R}(\alpha) 
    \\
+\frac{1}{2}\int_0^1 \Ll(\pi(s)\cdot \nabla\xi(\pi(s))-\xi(\pi(s)) \Rr)\d s .
\end{split}
\end{align}
Here, $\E$ integrates the Gaussian randomness in $w^{\nabla\xi\circ\pi}(\alpha)$ and then the randomness of $\mathfrak{R}$. It is noteworthy that $- \frac{1}{2}\nabla\xi\circ\pi(1)\cdot \tau\tau^\intercal = -\frac{1}{2}\E\Ll(w^{\nabla\xi\circ\pi}(\alpha)\cdot \tau\Rr)^2$, is the self-overlap correction in this functional.

We set
\begin{align}\label{e.sP(pi)=}
    \sP(\pi) = \sP^h(\pi,0),\quad\forall \pi\in\Pi
\end{align}
which is independent of $h$. This has appeared in~\eqref{e.limF^soc_N=}. We also define
\begin{align}\label{e.sP^h=}
    \sP^h(x) = \inf_{\pi\in\Pi} \sP^h(\pi,x),\quad\forall x\in \R^d,
\end{align}
and we denote the function $x\mapsto \sP^h(x)$ simply as $\sP^h$.

\subsection{Preliminaries}

Recall the definition of the mean magnetization $m_N$ in~\eqref{e.m=}.
Due to the assumption that $h$ is bounded, we have
\begin{align}\label{e.|m|<}
    \sup_{N\in\N} \Ll|m_N\Rr|<\infty.
\end{align}
Throughout, we write $\R_+=[0,\infty)$.
For $N\in\N$ and $(t,x) \in \R_+\times \R^d$, we set
\begin{gather}
    \cF_N(t,x) = \E \tilde \cF_N(t,x), \label{e.cF_N=}
    \\
    \tilde\cF_N(t,x) = \frac{1}{N} \log \int \exp\Ll(H_N(\sigma) - \frac{N}{2}\xi\Ll(\frac{\sigma\sigma^\intercal}{N}\Rr) + tN G\Ll(m_N\Rr)+Nx\cdot m_N\Rr)\d P_N(\sigma).\notag
\end{gather}
We denote by $\cF_N$ the function $(t,x)\mapsto \cF_N(t,x)$.

In this section, we always assume conditions~\ref{i.P_1}--\ref{i.h_G}. We start by identifying the limit of the initial value $\cF_N(0,\cdot)$.

\begin{lemma}\label{l.limcF_N(0,x)=}
For every $x\in\R^d$, we have
\begin{align*}
    \lim_{N\to\infty} \cF_N(0,x) =\sP^h(x).
\end{align*}
\end{lemma}
\begin{proof}
The convergence at $x=0$ has been proved in~\cite[Corollary~8.3]{HJ_critical_pts} (which is exactly~\eqref{e.limF^soc_N=}). For $x\neq 0$, we can substitute $\d P_1$ with $e^{x\cdot h(\tau)}\d P_1(\tau)$ and notice that the new measure still satisfies~\ref{i.P_1}. The convergence at $x$ follows from this substitution.
\end{proof}

For any function $g:(0,\infty)\times \R^d\to\R$, we denote the derivative of $g$ in the first variable by $\partial_tg$, the ($\R^d$-valued) gradient in the second variable by $\nabla g$, and the Laplacian in the second variable by $\Delta g$.

\begin{lemma}\label{l.sP}
The function $\sP^h:\R^d\to\R$ is convex, Lipschitz, and continuously differentiable.
\end{lemma}

\begin{proof}
For $x\in\R^d$, let $\la\cdot\ra_x$ be the Gibbs measure naturally associated with $\cF_N(0,x)$.
Fixing any $x,\, y \in \R^d$, for $r\in[0,1]$, we can compute
\begin{align}\label{e.dF_N(x+ry)}
    \frac{\d}{\d r}\cF_N(0,x+ry) = \E \la y\cdot m_N\ra_{x+ry}.
\end{align}
The uniform bound on $m_N$ as in~\eqref{e.|m|<} implies $\Ll|\frac{\d}{\d r}\cF_N(0,x+ry)\Rr|\leq C|y|$ uniformly in $x,y,r,N$ for some constant $C$. Hence, we have $|\cF_N(0,x+y)-\cF_N(0,y)|\leq C|y|$. Sending $N\to\infty$ and using Lemma~\ref{l.limcF_N(0,x)=}, we can deduce that $\sP^h$ is Lipschitz.

We denote by $m'_N$ an independent copy of $m_N$ under the Gibbs measure.
We differentiate~\eqref{e.dF_N(x+ry)} one more time to get
\begin{align}\label{e.2nd_der_cF}
\begin{split}
    \frac{\d^2}{\d r^2}\cF_N(0,x+ry) & = N\E \la \Ll(y\cdot m_N\Rr)^2 -\Ll(y\cdot m_N\Rr)\Ll(y\cdot m'_N\Rr) \ra_{x+ry}
    \\
    & = N\E \la \Ll(y\cdot m_N\Rr)^2 -\la y\cdot m_N\ra_{x+ry}^2 \ra_{x+ry}.
\end{split}
\end{align}
Since the right-hand side is nonnegative, we have verified that $\cF_N(0,\cdot)$ is convex. 
Sending $N\to\infty$ and using Lemma~\ref{l.limcF_N(0,x)=}, we obtain the convexity of $\sP^h$.

It is classical that (e.g.\ \cite[Theorem~25.5]{rockafellar1970convex})
the combination of convexity and differentiability implies continuous differentiability. Hence, it remains to show that $\sP^h$ is differentiable everywhere.
Fix any $x \in \R^d$. 
A vector $a\in \R^d$ is said to be a subdifferential of $\sP^h$ at $x$ if $\sP^h(y)-\sP^h(x)\geq a\cdot(y-x)$ for every $y\in\R^d$. 
Since $\sP^h$ is convex, it suffices to show that any subdifferential $a$ of $\sP^h$ at $x$ is unique. 
For each $\eps>0$, we choose $\pi_\eps$ to satisfy
\begin{align}\label{e.P(pi_eps,x)<P(x)+1/n}
    \sP^h(\pi_\eps,x) \leq \sP^h(x) +\eps.
\end{align}
Fix any $y\in \R^d$ and let $r\in (0,1]$.
Using the definition of the subdifferential, the fact that $\sP^h$ is an infimum, and \eqref{e.P(pi_eps,x)<P(x)+1/n}, we have
\begin{align*}
    y\cdot a &\leq \frac{\sP^h(x+ry) - \sP^h(x)}{r} \leq \frac{\sP^h(\pi_\eps,x+ry)-\sP^h(\pi_\eps,x)+\eps}{r},
    \\
    y\cdot a &\geq \frac{\sP^h(x) - \sP^h(x-ry)}{r} \geq \frac{\sP^h(\pi_\eps,x)-\sP^h(\pi_\eps,x-ry)-\eps}{r}.
\end{align*}
We can compute derivatives of $\sP^h(\pi_\eps,\cdot)$ at $x$ similarly as in~\eqref{e.dF_N(x+ry)} and~\eqref{e.2nd_der_cF} to see that they are bounded.
Hence, we can use Taylor's expansion of $\sP^h(\pi_\eps,\cdot)$ at $x$ in the above display to get
\begin{align*}
    \Ll|y\cdot a - y\cdot \nabla\sP^h(\pi_\eps,x)\Rr| \leq C r+\eps r^{-1}
\end{align*}
for some constant $C$ independent of $r$ and $\eps$. Setting $r= \sqrt{\eps}$ and sending $\eps\to0$, we can see that $y\cdot a$ is uniquely determined. Since $y$ is arbitrary, we conclude that the subdifferential of $\sP^h$ at $x$ is unique, and thus $\sP^h$ is differentiable at $x$. As explained earlier, this implies that $\sP^h$ is continuously differentiable.
\end{proof}

\begin{lemma}\label{l.cF_N_properties}
The following holds:
\begin{itemize}
    \item for each $N$, $\cF_N$ is Lipschitz and convex jointly over $\R_+\times \R^d$;
    \item the Lipschitzness is uniform in $N$, namely, $\sup_{N\in\N}\|\cF_N\|_\mathrm{Lip}<\infty$;
    \item there is a constant $C>0$ such that, everywhere on $(0,\infty)\times \R^d$ and for every $N$,
    \begin{align}\label{e.approx_hj}
        \Ll| \partial_t \cF_N - G\Ll(\nabla \cF_N\Rr)\Rr| \leq C\Ll(N^{-1} \Delta \cF_N\Rr)^\frac{1}{2} + C\E \Ll|\nabla\tilde \cF_N -\nabla \cF_N\Rr|.
    \end{align}
\end{itemize}
\end{lemma}

\begin{proof}
For brevity, we write $m=m_N$ and let $\la\cdot\ra$ be the Gibbs measure naturally associated with $\cF_N(t,x)$. The value of $(t,x)$, on which $\la\cdot\ra$ depends, will be clear from the context.
At every $(t,x) \in (0,\infty)\times \R^d$, we can compute
\begin{align}\label{e.1st_der_FN}
    \partial_t \cF_N =\E\la G(m)\ra,  \qquad \nabla \tilde\cF_N = \la m \ra, \qquad \nabla \cF_N = \E \la m \ra
\end{align}
and, for any $(s,y)\in\R\times \R^d$,
\begin{align}\label{e.2nd_der_FN}
    \frac{\d^2}{\d \eps^2}\tilde\cF_N(t+\eps s, x+\eps y)\,  \Big|_{\eps =  0} = N  \la \Ll(s G(m)+y\cdot m\Rr)^2 - \la s G(m)+y\cdot m \ra^2 \ra\geq 0. 
\end{align}
Since $|m|$ is bounded (as in~\eqref{e.|m|<}) and $G$ is locally Lipschitz, we deduce from~\eqref{e.1st_der_FN} that $\mathcal{F}_N$ is Lipschitz in both variables with $\sup_{N}\| \mathcal{F}_N\|_\mathrm{Lip}<\infty$.
Since it is easy to verify $\frac{\d^2}{\d \eps^2}\cF_N(t+\eps s, x+\eps y) = \E\frac{\d^2}{\d \eps^2}\tilde\cF_N(t+\eps s, x+\eps y)$, we obtain from~\eqref{e.2nd_der_FN} the convexity of $\cF_N$. Setting $s=0$ in~\eqref{e.2nd_der_FN}, we can get
\begin{align*}
    \Delta \cF_N = N  \E \la | m|^2 - |\la m\ra |^2 \ra = N  \E \la |m - \la m\ra |^2 \ra.
\end{align*}
Also, from~\eqref{e.1st_der_FN}, we have
\begin{align*}
    \Ll|\partial_t \cF_N - G(\nabla\cF_N)\Rr|
    &= \Ll|\E\la G(m)\ra - G\Ll(\E\la m\ra\Rr)\Rr|
    \\
    &\leq \Ll|\E\la G(m)\ra - \E G(\la m\ra) \Rr|  + \Ll|\E G(\la m\ra) -  G\Ll(\E\la m\ra\Rr)\Rr|
    \\
    &\leq C\E \la \Ll|m - \la m\ra \Rr|\ra + C\E \Ll|\la m\ra - \E \la m\ra\Rr|
\end{align*}
for some constant $C$ due to the local Lipschitzness of $G$;
and we have
\begin{align*}
    \E \Ll|\nabla \tilde \cF_N - \nabla \cF_N\Rr| = \E \Ll|\la m\ra -\E \la m\ra\Rr|.
\end{align*}
The above three displays together yield~\eqref{e.approx_hj}.
\end{proof}

\subsection{A PDE approach}
The estimate~\eqref{e.approx_hj} hints that the limit of $\cF_N$ should be the solution of
\begin{align*}
    \partial_t f- G(\nabla f)=0,\quad\text{on $(0,\infty)\times \R^d$},
\end{align*}
which is indeed the case. Here, the solution is understood in the viscosity sense. A function $f:\R_+\times \R^d\to\R$ is a \textit{viscosity subsolution} (respectively, \textit{supersolution}) of
\begin{align}\label{e.hj}
    \partial_t f- G(\nabla f) =0 \quad \text{on $\Ll(0,\infty\Rr)\times \R^d$}
\end{align}
if whenever there is a smooth $\phi:(0,\infty)\times \R^d\to\R$ such that $f-\phi$ achieves a local maximum (respectively, minimum) at some $(t,x) \in (0,\infty)\times \R^d$, we have $(\partial_t \phi-G(\nabla\phi))(t,x)\leq 0$ (respectively, $\geq 0$). If $f$ is both a viscosity subsolution and supersolution, we call $f$ a \textit{viscosity solution}.

If $f$ is a limit of $\cF_N$, then Lemma~\ref{l.limcF_N(0,x)=} implies that the relevant initial condition should be $f(0,\cdot) = \sP^h$, which is continuously differentiable by Lemma~\ref{l.sP}. 
Also, Lemma~\ref{l.cF_N_properties} ensures that $f$ is Lipschitz and convex. These two properties make the following \textit{convex selection principle} useful in this setting. We refer to \cite[Theorem~3.21 and Corollary~3.24]{HJbook} for the proof. This result first appeared in~\cite{chen2022statistical}.
\begin{theorem}[Convex selection principle]\label{t.convex_select}
Let $f:\R_+\times \R^d\to\R$ be a jointly convex and jointly Lipschitz continuous. Suppose that for any $(t,x)\in (0,\infty)\times \R^d$ and any smooth function $\phi:(0,\infty)\times\R^d\to\R$ such that $f-\phi$ has a strict local maximum at $(t,x)$, we have $(\partial\phi-G(\nabla \phi))(t,x)=0$. If moreover $f(0,\cdot)$ is continuously differentiable, then $f$ is the viscosity solution to~\eqref{e.hj}.
\end{theorem}

Here, the convergence in the local uniform topology means uniform convergence on every compact subset. As an application of this theorem, we can identify the limit of $\cF_N$.

\begin{lemma}\label{l.FN_cvg_f}
As $N\to\infty$, $\cF_N$ converges in the local uniform topology to the unique viscosity solution $f$ of~\eqref{e.hj} with initial condition $f(0,\cdot)=\sP^h$.
\end{lemma}

\begin{proof}
Since $\cF_N$ is Lipschitz uniformly in $N$ (by Lemma~\ref{l.cF_N_properties}), the Arzel\`a--Ascoli theorem implies that any subsequence of $\cF_N$ has a further subsequence that converges in the local uniform topology to some $f$. By the uniqueness of the solution (e.g.\ \cite[Corollary~3.7]{HJbook}), it suffices to show that any such $f$ is the viscosity solution. For lighter notation, we assume that the entire sequence $\cF_N$ converges to $f$.

Let $(t,x)\in(0,\infty)\times \R^d$ and smooth $\phi$ satisfy that $f-\phi$ has a strict local maximum at $(t,x)$. The goal is to show that 
\begin{equation}
\label{e.eq.phi}
 (\partial_t\phi-G(\nabla\phi))(t,x)= 0.
\end{equation}
Before showing~\eqref{e.eq.phi}, let us use this to deduce the announced result.
By Lemma~\ref{l.limcF_N(0,x)=}, we have $f(0,\cdot)=\sP^h$. Since $\sP^h$ is continuously differentiable due to Lemma~\ref{l.sP} and since Lemma~\ref{l.cF_N_properties} implies that $f$ is Lipschitz and convex, we are allowed to use \eqref{e.eq.phi} and Theorem~\ref{t.convex_select} to conclude that $f$ solves the equation in~\eqref{e.hj}.

It remains to verify~\eqref{e.eq.phi}.
By the local uniform convergence, there exists $(t_N,x_N) \in (0,\infty) \times \R^d$ such that $\cF_N - \phi$ has a local maximum at $(t_N,x_N)$, and $\lim_{N\to\infty}(t_N,x_N)=(t,x)$. Notice that 
\begin{equation}  
\label{e.gradient.equalities}
(\partial_t,\nabla)( \cF_N - \phi)(t_N,x_N) = 0 .
\end{equation}
Throughout this proof, we denote by $C < \infty$ a constant that may vary from one occurrence to the next and is allowed to depend on $(t,x)$ and~$\phi$. 

We want to show that, for every $y \in \R^d$ with $|y| \le C^{-1}$,
\begin{equation}
\label{e.hessian.est}
0 \le \cF_N(t_N,x_N + y) - \cF_N(t_N,x_N ) - y \cdot \nabla \cF_N(t_N,x_N) \le C |y|^2.
\end{equation}
The convexity of $\cF_N$ gives the first inequality. To derive the other, we start by using Taylor's expansion:
\begin{align}  
\label{e.taylor}
\begin{split}
&\cF_N(t_N,x_N + y) - \cF_N(t_N,x_N ) 
\\
&= y \cdot \nabla \cF_N(t_N,x_N) + \int_0^1 (1-s) y\cdot \nabla\Ll(y \cdot \nabla \cF_N\Rr)(t_N,x_N + s y) \, \d s.
\end{split}
\end{align}
The same holds with $\cF_N$ replaced by $\phi$. By the local maximality of $\cF_N - \phi$ at $(t_N,x_N)$,
\begin{equation*}  \cF_N(t_N,x_N + y) - \cF_N(t_N,x_N ) \le \phi(t_N,x_N + y) - \phi(t_N,x_N ) 
\end{equation*}
holds for every $|y| \le C^{-1}$.
The above two displays along with~\eqref{e.gradient.equalities} and Taylor's expansion of $\phi$ (similar to~\eqref{e.taylor}) imply
\begin{equation*}  \int_0^1 (1-s)  y\cdot \nabla\Ll(y \cdot \nabla \cF_N\Rr)(t_N,x_N + s y)  \, \d s \le \int_0^1 (1-s)  y\cdot \nabla\Ll(y \cdot \nabla \phi\Rr)(t_N,x_N + s y)  \, \d s.
\end{equation*}
Since the function $\phi$ is smooth, the right side of this inequality is bounded by $C |y|^2$. Using~\eqref{e.taylor} once more, we obtain~\eqref{e.hessian.est}. 

Next, setting $B = \Ll\{ (t',x') \in \R_+ \times \R^d \ : \ |t'-t| \le C^{-1} \text{ and } |x'-x| \le C^{-1} \Rr\}$ and $\delta_N = \E \Ll[ \sup_B \Ll|\tilde\cF_N - \cF_N\Rr|  \Rr]$, we show
\begin{equation}
\label{e.concentration}
\E \Ll[ \Ll|\nabla\tilde\cF_N - \nabla \cF_N\Rr|(t_N,x_N) \Rr] \le C \delta_N^\frac 1 2.
\end{equation}
Using the convexity of $\tilde\cF_N$ due to~\eqref{e.2nd_der_FN}, we have
\begin{equation*}  \tilde\cF_N(t_N,x_N+ y) \ge\tilde\cF_N(t_N,x_N) + y \cdot \nabla\tilde\cF_N(t_N,x_N).
\end{equation*}
Combining this with \eqref{e.hessian.est}, we obtain that, for every $|y| \le C^{-1}$,
\begin{equation*}  y \cdot \Ll(\nabla\tilde\cF_N - \nabla \cF_N\Rr)(t_N,x_N)\le 2 \sup_B \Ll|\tilde\cF_N - \cF_N\Rr| + C |y|^2 .
\end{equation*}
For some deterministic $\lambda \in [0,C^{-1}]$ to be determined, we fix the random vector
\begin{equation*}  y = \lambda \frac{\Ll(\nabla\tilde\cF_N - \nabla \cF_N\Rr)(t_N,x_N)}{\Ll|\nabla\tilde\cF_N - \nabla \cF_N\Rr|(t_N,x_N)},
\end{equation*}
to get
\begin{equation*}  \lambda \Ll|\nabla\tilde\cF_N - \nabla \cF_N\Rr|(t_N,x_N)  \le 2 \sup_B \Ll|\tilde\cF_N - \cF_N\Rr|  + C \lambda^2 .
\end{equation*}
By the standard concentration result (e.g.\ \cite[Theorem~1.2]{pan}) and an $\eps$-net to cover $B$, we can see $\lim_{\N\to\infty} \delta_N =0$. Taking the expectation in the above display and choosing $\lambda = \delta_N^\frac{1}{2}$, we obtain \eqref{e.concentration}. 

Since~\eqref{e.hessian.est} implies $|\Delta \cF_N(t_N,x_N)|\leq C$, using~\eqref{e.approx_hj}, \eqref{e.gradient.equalities}, and~\eqref{e.concentration}, we arrive at
\begin{align*}
    \Ll|\partial_t\phi - G(\nabla\phi)\Rr|(t_N,x_N)\leq CN^{-\frac{1}{2}} + C\delta^\frac{1}{2}_N.
\end{align*}
Sending $N\to\infty$ and using the convergence of $(t_N,x_N)$ to $(t,x)$, we get~\eqref{e.eq.phi}. As we explained previously, this completes the proof.
\end{proof}

Since $\sP^h$ is convex, the solution $f$ admits a variational representation.

\begin{proposition}[Hopf formula]\label{p.hopf}
At every $(t,x)\in\R_+\times \R^d$,
\begin{align*}
    \lim_{N\to\infty}\cF_N(t,x) = \sup_{z\in\R^d}\inf_{y\in\R^d}\Ll\{\sP^h(y)+z\cdot(x-y)+tG(z)\Rr\}.
\end{align*}
\end{proposition}

\begin{proof}
Let $f$ be the limit of $\cF_N$ given by Lemma~\ref{l.FN_cvg_f}. The convexity of $\sP^h$ proved in Lemma~\ref{l.sP} gives the convexity of $f(0,\cdot)$. This allows us to represent $f$ in terms of the Hopf formula \cite{bardi1984hopf,lions1986hopf}. We refer to \cite[Theorem~3.13]{HJbook} for the version needed here.
\end{proof}

\subsection{Proof of main results}

Now, we are ready to prove the main theorem and its corollaries from Section~\ref{s.intro}

\begin{proof}[Proof of Theorem~\ref{t.F^soc,G_N}]
Comparing $F_N^{\soc,G}$ in~\eqref{e.F^soc,G_N=} and $\cF_N$ in~\eqref{e.cF_N=}, we have $F_N^{\soc,G}= \cF_N(1,0)$.
Then, the convergence in~\eqref{e.limF^soc,G_N=} follows from Proposition~\ref{p.hopf} for $(t,x)=(1,0)$ and the definition of $\sP^h$ in~\eqref{e.sP^h=}.
\end{proof}

\begin{proof}[Proof of Corollary~\ref{c.F^G_N}]
Recall the function $\bfs$ defined in~\eqref{e.bfs} and that we can identify $\S^\D$ with $\R^{\D(\D+1)/2}$ as described above~\eqref{e.bfs}. Write $X=\S^\D\times \R^d$ which is isometric to $\R^{\frac{\D(\D+1)}{2}+d}$. Consider the function $\bar h:\R^\D\to X$ given by $\bar h:\:\tau\mapsto (\bfs(\tau),h(\tau))$ and the function $\bar G:X\to \R$ given by $\bar G:  (z,m)\mapsto  \frac{1}{2}\xi(z)+G(m)$. Then, we can see that $F^G_N$ in~\eqref{e.F^G_N=} is equal to $F^{\soc,G}_N$ in~\eqref{e.F^soc,G_N=} with $\R^d$, $h$, $G$ therein substituted with $X$, $\bar h$, $\bar G$. Hence, the convergence in~\eqref{e.limF^G_N=} is given by that in~\eqref{e.limF^soc,G_N=} with this substitution.
\end{proof}

\begin{proof}[Proof of Corollary~\ref{c.F^G_N_so}]
Under the assumption $h=\bfs$, we can see that $F^G_N$ in~\eqref{e.F^G_N=} is equal to $F^{\soc,G}_N$ in~\eqref{e.F^soc,G_N=} with $\R^d$, $h$, $G$ therein substituted with $\S^\D$, $\bfs$, $\frac{1}{2}\xi+G$.
Then,~\eqref{e.limF^G_N=so} follows from~\eqref{e.limF^soc,G_N=}. Lastly, $F_N$ in~\eqref{e.F_N=} is equal to $F^G_N$ with $G=0$ and thus~\eqref{e.limF_N=} easily follows from~\eqref{e.limF^G_N=so}.
\end{proof}

\section{Approach via constrained free energy}\label{s.alt_proof}

In a standard vector spin glass model without the self-overlap correction, it is necessary to consider the self-overlap as a conventional order parameter. These models have been rigorously studied by Panchenko in \cite{pan05,pan.potts,pan.vec}. Aside from the machinery already needed for the SK model as in \cite{pan}, the additional strategy is to consider free energy with a constraint on the self-overlap and then argue along the lines of the large deviations theory. 
One can rework these arguments on the free energy with a constraint on a general conventional order parameter (e.g.\ the limit of $m_N$ in~\eqref{e.m=}) to prove Theorem~\ref{t.F^soc,G_N} and its corollaries. Compared to the PDE approach presented in Section~\ref{s.proofs}, one needs to modify these arguments from the very beginning so that the full presentation can be lengthy.

When the conventional order parameter is solely the self-overlap, the modification is minimal. We choose to present the approach in this special case by proving Theorem~\ref{t.constraint_approach} below using results from~\cite{pan.vec}. Results in~\cite{pan.vec} only allow us to prove the theorem under a stronger assumption that $\xi$ is convex over $\R^{\D\times\D}$ instead of $\S^\D_+$ in~\ref{i.xi_convex}. We will explain where the stronger convexity is needed. As a substitute, we use Lemma~\ref{l.limcF_N(0,x)=}.

Recall the function $\bfs$ in~\eqref{e.bfs}.
Let
\begin{align}\label{e.cK=}
    \text{$\cK$ be the closed convex hull of $\Ll\{\tau\tau^\intercal:\: \tau \in \supp P_1\Rr\}$. }
\end{align}
The following theorem slightly refines Corollary~\ref{c.F^G_N_so} as $\sup_z$ is now taken over $\cK$ instead of $\S^\D$.
\begin{theorem}\label{t.constraint_approach}
Under conditions~\ref{i.P_1}--\ref{i.h_G} and an additional assumption that $h=\bfs$ (identifying $\S^\D$ with $\R^{\D(\D+1)/2}$ isometrically), the limit of $F^{G}_N$ in~\eqref{e.F^G_N=} is given by
\begin{align}\label{e.limF^G_N=so_K}
    \lim_{N\to\infty} F^G_N = \sup_{z\in\cK}\inf_{\pi\in\Pi}\inf_{x\in \S^\D}\Ll\{\sP^{\bfs}(\pi,x)-z\cdot x +\frac{1}{2}\xi(z)+G(z)\Rr\}.
\end{align}
\end{theorem}

The remainder of this section is devoted to the proof of this theorem. Notice that, when $h=\bfs$, $m_N$ in~\eqref{e.m=} becomes the self-overlap, namely,
\begin{align}\label{e.m=so}
    m_N = \frac{\sigma\sigma^\intercal}{N}.
\end{align}
For $N\in\N$ and a measurable subset $S\subset \R^{\D\times N}$, we consider
\begin{align}\label{e.F_N^G(S)=}
    F_N^G(S) = \frac{1}{N}\E\log \int_S\exp\Ll(H_N(\sigma) + NG\Ll(m_N\Rr)\Rr)\d P_N(\sigma),
\end{align}
which is a constrained version of $F^G_N$ in~\eqref{e.F^G_N=} with $h$ therein equal to $\bfs$ in~\eqref{e.bfs}.
We set $F^0_N(S)$ to be the constrained free energy with $G$ in~\eqref{e.F_N^G(S)=} set to be zero. Notice that $F_N^G$ in~\eqref{e.F^G_N=} is now equal to $F_N^G\Ll(\R^{\D\times N}\Rr)$.
For every $z\in \S^\D$ and $\eps>0$, we define
\begin{align}\label{e.Sigma(z)}
    \Sigma_\eps(z) = \Ll\{\sigma\in \R^{\D\times N}:\: \Ll|m_N-z \Rr|\leq \eps\Rr\}.
\end{align}
Recall the definition of $\sP^\bfs(y)$ in~\eqref{e.sP^h=} with $\bfs$ substituted for $h$.
We start with a result on the limit of the constrained free energy.

\begin{proposition}\label{p.lim_const_F_N}
Under conditions~\ref{i.P_1}--\ref{i.h_G} and an additional assumption that $h=\bfs$, it holds for every $z\in\cK$ that
\begin{align}\label{e.liminf=limsup}
    \lim_{\eps\downarrow0}\liminf_{N\to\infty} F^0_N\Ll(\Sigma_\eps(z)\Rr) = \lim_{\eps\downarrow0}\limsup_{N\to\infty} F^0_N\Ll(\Sigma_\eps(z)\Rr) = \inf_{y\in \S^\D} \Ll\{\sP^\bfs(y)-y\cdot z+\frac{1}{2}\xi(z)\Rr\}.
\end{align}
\end{proposition}

We remark that, to our best knowledge, this proposition is new because the convexity of $\xi$ is only assumed to be over $\S^\D_+$ as in~\ref{i.xi_convex}.
Under a stronger assumption that $\xi$ is convex on $\R^{\D\times \D}$, this proposition is direct consequence of \cite[Theorem~2 and Lemma~2]{pan.vec}. The stronger convexity is needed in the proof of the upper bound \cite[Lemma~2]{pan.vec} to use Guerra's interpolation.

\begin{proof}
It is proven in \cite[Theorem~2]{pan.vec} that
\begin{align*}
    \lim_{\eps\downarrow0}\liminf_{N\to\infty} F_N^0\Ll(\Sigma_\eps(z)\Rr) \geq \inf_{y\in \S^\D}\inf_{\pi:\:\pi(1)=z} \Ll\{\sP^\bfs(\pi,y)-y\cdot z+\frac{1}{2}\xi(z)\Rr\}.
\end{align*}
Since $\inf_{\pi:\:\pi(1)=z}\sP^\bfs(\pi,y)\geq \inf_{\pi}\sP^\bfs(\pi,y)=\sP^\bfs(y)$, we get the lower bound for~\eqref{e.liminf=limsup}.
The Parisi functional in \cite[Theorem~2]{pan.vec} is written in a different notation. It is explained in~\cite[Appendix~A]{chen2023on} how to rewrite it in our notation.

Now, we focus on the upper bound. Let $y\in \S^\D$ and recall the expression of $\cF_N(0,y)$ in~\eqref{e.cF_N=}. By the assumption $h=\bfs$ and the consequence~\eqref{e.m=so}, we have
\begin{align*}
    \cF_N(0,y) =  \frac{1}{N} \log \int \exp\Ll(H_N(\sigma) - \frac{N}{2}\xi\Ll(m_N\Rr) +Ny\cdot m_N\Rr)\d P_N(\sigma).
\end{align*}
Then, we can compute
\begin{align*}
    F^0_N\Ll(\Sigma_\eps(z)\Rr) - \cF_N(0,y) + y\cdot z -\frac{1}{2}\xi(z) = \frac{1}{N}\E \log \la \mathds{1}_{\sigma\in \Sigma_\eps(z)} e^{\frac{N}{2}\xi\Ll(m_N\Rr)-Ny\cdot m_N+Ny\cdot z -\frac{N}{2}\xi(z) }\ra
\end{align*}
where $\la \cdot\ra$ is the Gibbs measure naturally associated with $\cF_N(0,y)$. Notice that $m_N\in\cK$. Let $L$ be the Lipschitz coefficient of $\xi$ on $\cK$. Then, we can use the definition of $\Sigma_\eps(z)$ in~\eqref{e.Sigma(z)} to see that the right-hand side in the display is bounded from above by $\frac{1}{2}L\eps + |z|\eps$. This along with Lemma~\ref{l.limcF_N(0,x)=} (with $h=\bfs$) implies that
\begin{align*}
    \limsup_{N\to\infty} F_N^0\Ll(\Sigma_\eps(z)\Rr)\leq \sP^\bfs(y)-y\cdot z + \frac{1}{2}\xi(z)+\frac{1}{2}L\eps+|z|\eps.
\end{align*}
Sending $\eps\downarrow0$, we get the desired upper bound for~\eqref{e.liminf=limsup}.
\end{proof}

\begin{proof}[Proof of Theorem~\ref{t.constraint_approach}]
First, we show that there is a constant $C$ such that, for every $\eps>0$ and $z\in\cK$,
\begin{align}\label{e.F^G-F}
    \Ll|F^G_N\Ll(\Sigma_\eps(z)\Rr) - F^0_N\Ll(\Sigma_\eps(z)\Rr) - G(z)\Rr|\leq C\eps.
\end{align}
To see this, we compute
\begin{align*}
    F^G_N\Ll(\Sigma_\eps(z)\Rr) - F^0_N\Ll(\Sigma_\eps(z)\Rr) - G(z) = \frac{1}{N}\E \log \la e^{NG\Ll(m_N\Rr)- NG(z)} \ra
\end{align*}
where $\la\cdot\ra$ is the Gibbs measure associated with $F^0_N\Ll(\Sigma_\eps(z)\Rr)$. Due to the constraint imposed by $\Sigma_\eps(z)$, we have $\Ll|m_N-z\Rr|\leq \eps$ a.s.\ under $\la\cdot\ra$. Letting $C$ be the Lipschitz coefficient of $G$ on $\cK$, we can obtain~\eqref{e.F^G-F}.

Next, we show the lower bound for~\eqref{e.limF^G_N=so_K}.
For brevity, we write
\begin{align*}
    \mathscr{Q}(z) =\inf_{y\in \S^\D} \Ll\{\sP^\bfs(y)-y\cdot z+\frac{1}{2}\xi(z) +G(z)\Rr\},\quad\forall z\in\cK.
\end{align*}
Also, recall that we can expand $\sP^\bfs(y) = \inf_{\pi\in\Pi}\sP^\bfs(\pi,y)$ as defined in~\eqref{e.sP^h=}.
Since $F^G_N$ in~\eqref{e.F^G_N=} does not have any constraint, it is easy to see $F^G_N\Ll(\Sigma_\eps(z)\Rr)\leq F^G_N$, which along with~\eqref{e.F^G-F} implies 
\begin{align*}
    \liminf_{N\to\infty}F^G_N \geq \liminf_{N\to\infty} F^0_N\Ll(\Sigma_\eps(z)\Rr) +G(z)- C\eps
\end{align*}
for every $z\in \cK$ and $\eps>0$. Sending $\eps\downarrow0$ and using Proposition~\ref{p.lim_const_F_N}, we get
\begin{align*}
    \liminf_{N\to\infty}F^G_N \geq\mathscr{Q}(z).
\end{align*}
Taking the supremum over $z\in\cK$, we obtain the lower bound for~\eqref{e.limF^G_N=so_K}.

The proof for the upper bound is contained in the proof of \cite[Lemma~3]{pan.vec}. For completeness, we present it below.
Temporarily fix any $\delta>0$. By Proposition~\ref{p.lim_const_F_N} and~\eqref{e.F^G-F}, for every $z\in\cK$, there is $\eps_z>0$ such that
\begin{align}\label{e.limsupF^G_N<}
    \limsup_{N\to\infty} F^G_N(\Sigma_{\eps_z}(z))\leq \delta+ \mathscr{Q}(z).
\end{align}
Since $\cK$ is compact, there are an integer $n\in\N$ and $z_1,\dots,z_n\in \cK$ such that $\cK$ is covered by $\eps_{z_i}$-balls centered at $z_i$ for $1\leq i\leq n$. 
Since $P_N$ is a product measure, using the definition of $\Sigma_\eps(z)$ in~\eqref{e.Sigma(z)} and that of $\cK$ in~\eqref{e.cK=}, we have
\begin{align}\label{e.suppP_N_subset}
    \supp P_N\subset \cup_{i=1}^n\Sigma_{\eps_{z_i}}(z_i).
\end{align}

For $N\in\N$ and a measurable subset $S\subset \R^{\D\times N}$, we define
\begin{align*}
    \tilde F_N^G(S) = \frac{1}{N}\log \int_{S} \exp\Ll(H_N(\sigma)+NG\Ll(m_N\Rr)\Rr)\d P_N(\sigma)
\end{align*}
and set $\tilde F^G_N=\tilde F^G_N\Ll(\R^{\D\times N}\Rr)$. So, we have $ F^G_N(S) = \E \tilde F^G_N(S)$ and $F^G_N = \E \tilde F^G_N$. Due to $\tilde F^G_N= \tilde F^G_N\Ll(\supp P_N\Rr)$, we can use~\eqref{e.suppP_N_subset} to get
\begin{align*}
    \tilde F^G_N \leq \frac{1}{N}\log \Ll(n\max_{1\leq i\leq n}e^{N \tilde F^G_N\Ll(\Sigma_{\eps_{z_i}}(z_i)\Rr)}\Rr) \leq N^{-1}\log n + \max_{1\leq i\leq n} \tilde F^G_N\Ll(\Sigma_{\eps_{z_i}}(z_i)\Rr).
\end{align*}
The standard Gaussian concentration inequalities (e.g.\ \cite[Theorem~1.2]{pan}) gives a constant $C$ such that $\E \Ll|\tilde F^G_N(S) - F^G_N(S)\Rr|\leq CN^{-\frac{1}{2}}$ uniformly in $N$ and $S$. Therefore, we get
\begin{align*}
    F^G_N \leq N^{-1}\log n + 2CN^{-\frac{1}{2}}+\max_{1\leq i\leq n}  F^G_N\Ll(\Sigma_{\eps_{z_i}}(z_i)\Rr).
\end{align*}
Sending $N\to\infty$ and using~\eqref{e.limsupF^G_N<}, we obtain
\begin{align*}
    \limsup_{N\to\infty} F^G_N \leq \delta + \sup_{z\in\cK}\mathscr{Q}(z).
\end{align*}
The upper bound for~\eqref{e.limF^G_N=so_K} follows by taking $\delta\to0$.
\end{proof}

\appendix

\section{Convergence of \texorpdfstring{$m_N$}{the mean magnetization}}\label{s.cvg_m}
As mentioned in Remark~\ref{r.m_N}, we can show that $m_N$ in~\eqref{e.m=} always converges under the Gibbs measure associated with $F^\soc_N$ in~\eqref{e.F^soc_N=}. When $h=\bfs$, we have that $m_N = \frac{\sigma\sigma^\intercal}{N}$ is the self-overlap and such a result has been proved in~\cite[Theorem~1.1 (1) and (2)]{chen2023on}. A straightforward modification gives the desired result below.

\begin{proposition}\label{p.cvg_m}
Under conditions~\ref{i.P_1}--\ref{i.xi_convex}, if $h$ is bounded and measurable, then $m_N$ in~\eqref{e.m=} satisfies
\begin{align*}
    \lim_{N\to\infty} \E \la\Ll|m_N - \nabla \sP^h(0)\Rr| \ra = 0
\end{align*}
where $\sP^h$ is defined in~\eqref{e.sP^h=} and $\la\cdot\ra$ is the Gibbs measure associated with $F^\soc_N$ in~\eqref{e.F^soc_N=}.
\end{proposition}

For completeness, we present the proof, which follows from the straightforward combination of the next two lemmas. We assume~\ref{i.P_1}--\ref{i.xi_convex} henceforth.

\begin{lemma}
Let $\la\cdot\ra$ be associated with $F^\soc_N$.
If $h$ is bounded and measurable, then
\begin{align*}
    \lim_{N\to\infty}\E \la m_N\ra  = \nabla \sP^h(0).
\end{align*}
\end{lemma}

\begin{proof}
Recall $\cF_N$ defined in~\eqref{e.cF_N=}.
Let $y\in\R^d$ and $r>0$.
The convexity of $\cF_N$ by Lemma~\ref{l.cF_N_properties} implies
\begin{align*}
    \frac{\cF_N(0,0)-\cF_N(0,-ry)}{r} \leq y\cdot\nabla \cF_N(0,0) \leq \frac{\cF_N(0,ry)-\cF_N(0,0)}{r}.
\end{align*}
Sending $N\to\infty$ and then $r\to0$, and using Lemma~\ref{l.limcF_N(0,x)=} and the differentiability of $\sP^h$ in Lemma~\ref{l.sP}, we get
\begin{align}\label{e.cvgEnablaF_N}
    \lim_{N\to\infty} y\cdot\nabla \cF_N(0,0) = y\cdot\nabla \sP^h(0).
\end{align}
Varying $y$, we get $\lim_{N\to\infty} \nabla \cF_N(0,0) = \nabla \sP^h(0)$ in $\R^d$. Recall from~\eqref{e.1st_der_FN} that $\nabla \cF_N(0,0) = \E \la m_N\ra$ where $\la\cdot\ra$ is associated with $\cF_N(0,0)$. The desired result follows from the observation that $\cF_N(0,0)=F^\soc_N$.
\end{proof}

\begin{lemma}
Let $\la\cdot\ra$ be associated with $F^\soc_N$.
If $h$ is bounded and measurable, then
\begin{align*}
    \lim_{N\to\infty}\E\la\Ll|m_N - \E \la m_N\ra\Rr| \ra.
\end{align*}
\end{lemma}

\begin{proof}
For $x\in\R^d$, we write $\cF_N(x)=\cF_N(0,x)$ (in~\eqref{e.cF_N=}) for brevity. Let $\la\cdot\ra_x$ be the Gibbs measure associated with $\cF_N(x)$. Since $\cF_N(0) = F^\soc_N$, we have $\la\cdot\ra=\la\cdot\ra_0$.
Fix any $y\in \R^d$ and set $g(\sigma) = Ny\cdot m_N = y\cdot\sum_{i=1}^Nh(\sigma_i)$. It suffices to show
\begin{align}\label{e.E<|g-E<g>|>=0}
    \lim_{N\to\infty} \frac{1}{N}\E \la \Ll|g(\sigma) - \E \la g(\sigma)\ra_0\Rr|  \ra_0 =0.
\end{align}

\textit{Step~1.}
We show
\begin{align}\label{e.E<|g-<g>|>=0}
    \lim_{N\to\infty} \frac{1}{N}\E \la \Ll|g(\sigma) - \la g(\sigma)\ra_0\Rr|  \ra_0 =0.
\end{align} 
We denote by $(\sigma^l)_{l\in\N}$ independent copies of $\sigma$ under the relevant Gibbs measure. Let $r>0$. Integrating by parts, we get
\begin{align*}
    &r\E \la \Ll|g\Ll(\sigma^1\Rr) -g\Ll(\sigma^2\Rr)\Rr|\ra_0 
    \\
    &= \int_0^r \E \la \Ll|g\Ll(\sigma^1\Rr) -g\Ll(\sigma^2\Rr)\Rr|\ra_{sy} \d s
    -\int_0^r \int_0^t \frac{\d}{\d s} \E \la \Ll|g\Ll(\sigma^1\Rr) -g\Ll(\sigma^2\Rr)\Rr|\ra_{sy}\d s \d t.
\end{align*}
The integrand in the last term can be estimated as follows 
\begin{align*}
    \frac{\d}{\d s} \E \la \Ll|g\Ll(\sigma^1\Rr) -g\Ll(\sigma^2\Rr)\Rr|\ra_{sy}  =  \E \la \Ll|g\Ll(\sigma^1\Rr) -g\Ll(\sigma^2\Rr)\Rr|\Ll(g\Ll(\sigma^1\Rr)+g\Ll(\sigma^2\Rr)-2g\Ll(\sigma^3\Rr)\Rr)\ra_{sy}
    \\
    \geq -2 \E \la \Ll|g\Ll(\sigma^1\Rr) -g\Ll(\sigma^2\Rr)\Rr|^2\ra_{sy}\geq -8 \E \la \Ll|g\Ll(\sigma\Rr) -\la g\Ll(\sigma\Rr)\ra_{sy}\Rr|^2\ra_{sy}.
\end{align*}
Inserting this into the previous display, we obtain
\begin{align*}
    &\E \la \Ll|g\Ll(\sigma^1\Rr) -g\Ll(\sigma^2\Rr)\Rr|\ra_0 
    \\
    &\leq \frac{1}{r}\int_0^r\E \la \Ll|g\Ll(\sigma^1\Rr) -g\Ll(\sigma^2\Rr)\Rr|\ra_{sy} \d s
    + \frac{8}{r}\int_0^r \int_0^t \E \la \Ll|g\Ll(\sigma\Rr) -\la g\Ll(\sigma\Rr)\ra_{sy}\Rr|^2\ra_{sy}\d s \d t
    \\
    &\leq \frac{2}{r}\int_0^r\E \la \Ll|g\Ll(\sigma\Rr) -\la g\Ll(\sigma\Rr)\ra_{sy}\Rr|\ra_{sy} \d s
    + 8\int_0^r \E \la \Ll|g\Ll(\sigma\Rr) -\la g\Ll(\sigma\Rr)\ra_{sy}\Rr|^2\ra_{sy}\d s .
\end{align*}
Setting $\eps_N = \frac{1}{N} \int_0^r \E \la \Ll|g\Ll(\sigma\Rr) -\la g\Ll(\sigma\Rr)\ra_{sy}\Rr|^2\ra_{sy}\d s$, we can rewrite the above estimate as
\begin{align*}
    \frac{1}{N}\E \la \Ll|g\Ll(\sigma^1\Rr) -g\Ll(\sigma^2\Rr)\Rr|\ra_0 \leq 2 \sqrt{\frac{\eps_N}{rN}}+ 8 \eps_N.
\end{align*}
Using~\eqref{e.2nd_der_cF}, we have
\begin{align*}
    \eps_N & = \int_0^r \frac{\d^2}{\d s^2}\cF_N(sy) \d s =  y\cdot\nabla \cF_N(ry) - y \cdot\nabla \cF_N(0) 
    \\
    &\leq \frac{\cF_N((r+t)y) - \cF_N(ry)}{t} - \frac{\cF_N(0) - \cF_N(-t y)}{t}
\end{align*}
for any $t>0$, where the last inequality follows from the convexity of $\cF_N$ given by Lemma~\ref{l.cF_N_properties}.
Combining the above two displays, using Lemma~\ref{l.limcF_N(0,x)=}, and noticing $\sup_N\eps_N<\infty$ (due to~\eqref{e.|m|<}), we obtain
\begin{align*}
    \limsup_{N\to\infty} \frac{1}{8N}\E \la \Ll|g\Ll(\sigma^1\Rr) -g\Ll(\sigma^2\Rr)\Rr|\ra_0 \leq \frac{\sP^h((r+t)y) - \sP^h(ry)}{t} - \frac{\sP^h(0) - \sP^h(-t y)}{t}.
\end{align*}
We first send $r\to0$ and then $t\to 0$. Since $\sP^h$ is differentiable by Lemma~\ref{l.sP}, the right-hand side vanishes, which yields~\eqref{e.E<|g-<g>|>=0}. 

\textit{Step~2.}
We show
\begin{align}\label{e.E|<g>-E<g>|=0}
    \lim_{N\to\infty} \frac{1}{N}\E\Ll| \la g(\sigma)\ra_0 - \E\la g(\sigma)\ra_0\Rr|  =0.
\end{align}
Recall $\tilde \cF_N$ below~\eqref{e.cF_N=} and we write $\tilde \cF_N(\cdot)= \tilde\cF_N(0,\cdot)$ for brevity.
We can use~\eqref{e.1st_der_FN} to rewrite
\begin{align}\label{e.1/NE|g-g|=...}
     \frac{1}{N}\E\Ll| \la g(\sigma)\ra_0 - \E\la g(\sigma)\ra_0\Rr| = \E \Ll| y\cdot\nabla \tilde\cF_N(0) - y\cdot\nabla \cF_N(0)\Rr|.
\end{align}
For $r\in(0,1]$, we define
\begin{align*}
    \delta_N(r) = \Ll|\tilde\cF_N(-ry)-\cF_N(-ry)\Rr| + \Ll|\tilde\cF_N(0)-\cF_N(0)\Rr| + \Ll|\tilde\cF_N(ry)-\cF_N(ry)\Rr|.
\end{align*}
In view of~\eqref{e.2nd_der_FN}, $\tilde\cF_N$ is convex, which implies
that $y\cdot\nabla \tilde\cF_N(0) - y\cdot\nabla \cF_N(0)$ is bounded from above by
\begin{align*}
    \frac{\tilde\cF_N(ry) - \tilde\cF_N(0)}{r} - y\cdot \nabla \cF_N(0) \leq \frac{\cF_N(ry)-\cF_N(0)}{r}-y\cdot \nabla \cF_N(0)+ \frac{\delta_N(r)}{r}
\end{align*}
and from below by
\begin{align*}
    \frac{\tilde\cF_N(0) - \tilde\cF_N(-ry)}{r} - y\cdot \nabla \cF_N(0) \geq \frac{\cF_N(0)-\cF_N(-ry)}{r}-y\cdot \nabla \cF_N(0)- \frac{\delta_N(r)}{r}.
\end{align*}
By the standard concentration result (e.g.\ see \cite[Theorem~1.2]{pan}), there is a constant $C>0$ such that $\sup_{r\in(0,1]}\E \delta_N(r)\leq CN^{-\frac{1}{2}}$.
This along with Lemma~\ref{l.limcF_N(0,x)=} and~\eqref{e.cvgEnablaF_N} gives
\begin{align*}
    &\limsup_{N\to\infty} \E \Ll|y\cdot\nabla \tilde\cF_N(0) - y\cdot\nabla \cF_N(0)\Rr| 
    \\
    &\leq \Ll|\frac{\sP^h(ry)-\sP^h(0)}{r}-y\cdot \nabla \sP^h(0)\Rr| 
    + \Ll|\frac{\sP^h(0)-\sP^h(-ry)}{r}-y\cdot \nabla \sP^h(0)\Rr|.
\end{align*}
Sending $r\to\infty$, using the differentiability of $\sP^h$, and inserting this to \eqref{e.1/NE|g-g|=...}, we get~\eqref{e.E|<g>-E<g>|=0}.

In conclusion, \eqref{e.E<|g-E<g>|>=0} follows from \eqref{e.E<|g-<g>|>=0} and~\eqref{e.E|<g>-E<g>|=0} and thus the proof is complete.
\end{proof}

\noindent
\textbf{Data availability.}
No datasets were generated during this work.

\noindent
\textbf{Conflict of interest.}
The author has no conflicts of interest to declare.

\small
\bibliographystyle{abbrv}
\newcommand{\noop}[1]{} \def\cprime{$'$}

\end{document}